\documentclass[aps,pra,twocolumn,superscriptaddress,floatfix,nofootinbib,showpacs,longbibliography]{revtex4-1}
\pdfoutput=1
\usepackage[utf8]{inputenc}  
\usepackage[T1]{fontenc}     
\usepackage[british]{babel}  
\usepackage[colorlinks=true, citecolor=blue, urlcolor=blue]{hyperref}  
\usepackage{graphicx} 
\usepackage[babel]{microtype}  
\usepackage{amsmath,amssymb,amsthm,bm,amsfonts,mathrsfs,bbm} 

\usepackage{xspace}  
\usepackage{pgf,tikz}
\usepackage{xcolor}
\usepackage{multirow}
\usepackage{array}
\usepackage{bigstrut}
\usepackage{braket}
\usepackage{color}
\usepackage{natbib}
\usepackage{multirow}
\usepackage{mathtools}
\usepackage[normalem]{ulem}
\usepackage{float}
\usepackage[caption = false]{subfig}
\usepackage{xcolor,colortbl}
\usepackage{color}
\usepackage{tikz}
\usetikzlibrary{positioning,arrows.meta,calc}
\usepackage{amssymb}
\newcommand{\Tr}{\operatorname{Tr}}

\newcommand{\be}{\begin{equation}}
\newcommand{\ee}{\end{equation}}
\newcommand{\ba}{\begin{eqnarray}}
\newcommand{\ea}{\end{eqnarray}}
\newcommand{\ketbra}[2]{|#1\rangle \langle #2|}

\newtheorem{theorem}{Theorem}
\newtheorem{corollary}{Corollary}
\newtheorem{definition}{Definition}
\newtheorem{proposition}{Proposition}
\newtheorem{observation}{Observation}

\newtheorem{lemma}{Lemma}
\newenvironment{theorem'}
 {\expandafter\def\expandafter\thetheorem\expandafter{\thetheorem'}\theorem}
 {\endtheorem}
\usepackage{mathtools}

\begin{document}

\title{Local Inaccessibility of Random Classical Information and Their Implications in the Change Point Problem} 
\author{Snehasish Roy Chowdhury}
\affiliation{Physics and Applied Mathematics Unit, Indian Statistical Institute, 203 B.T. Road, Kolkata 700108, India.}
\author{Subhendu B. Ghosh}
\email{subhendubghosh@gmail.com}
\affiliation{Department of Physics of Complex Systems, S. N. Bose National Center for Basic Sciences,
Block JD, Sector III, Salt Lake, Kolkata 700106, India}
\author{Tathagata Gupta}
\affiliation{Physics and Applied Mathematics Unit, Indian Statistical Institute, 203 B.T. Road, Kolkata 700108, India.}
\author{Anandamay Das Bhowmik}
\affiliation{School of Physics, IISER Thiruvananthapuram, Vithura, Kerala 695551, India}
\author{Sutapa Saha}
\affiliation{Harish-Chandra Research Institute, HBNI, Chhatnag Road, Jhunsi, Allahabad 211 019, India}
\author{Some Sankar Bhattacharya}
\affiliation{Física Teòrica: Informació i Fenòmens Quàntics, Universitat Autònoma de Barcelona, 08193 Bellaterra, Spain}
\author{Tamal Guha}
\email{g.tamal91@gmail.com}
\affiliation{QICI Quantum Information and Computation Initiative, School of Computing and Data Sciences, The University of Hong Kong, Pokfulam Road, Hong Kong}
\affiliation{Cryptology and Security Research Unit, Indian Statistical Institute, 203 B.T. Road, Kolkata 700108, India}

\begin{abstract} 
Discrimination of quantum states under local operations and classical communication (LOCC) is an intriguing question in the context of local retrieval of classical information, encoded in the multipartite quantum systems. All the local quantum state discrimination premises, considered so far, mimic a basic communication set-up, where the spatially separated decoding devices are independent of any additional input. Here, exploring a generalized communication scenario, we introduce a framework for input-dependent local quantum state discrimination, which we call \textit{local random authentication} (LRA). We report that impossibility of LRA certifies the presence of entangled states in the ensemble, a feature absent from erstwhile nonlocality arguments based on local state discrimination. Additionally, we explore the salient features of this state discrimination prototype for arbitrary set of orthogonal quantum states and compare them with the traditional notion of local quantum state discrimination. Finally, our results reveal a fundamental information-theoretic implications in the local estimation of quantum change point problems.
\end{abstract}
\maketitle
\section{Introduction}
Encoding classical information in quantum systems offers significant advantages over classical schemes, due to the existence of entanglement both in the preparation and the measurement devices \cite{holevo1979capacity, fuchs1997nonorthogonal, holevo1998capacity, hastings2009superadditivity, chiribella2025communication}. However, causal constraints at the recievers' end---such as retrieval of the encoded information locally (by multiple spatially separated receivers)---restrict the preparation of encoding quantum systems from which information can be decoded reliably. This follows from the fact that not every set of multipartite orthogonal quantum states can be discriminated perfectly under local operations and classical communication (LOCC) \cite{bennett1999quantum, walgate2000local, ghosh2001distinguishability, walgate2002nonlocality, ghosh2004distinguishability, watrous2005bipartite, duan2007distinguishing, bandyopadhyay2011locc, chitambar2013local, banik2019multicopy, Yuan2022, wang2025detectors}. In particular, such a restriction has implications in context of classical communication via quantum channel, assisted by a friendly environment \cite{gregoratti2003quantum, hayden2005correcting, watrous2005bipartite, duan2009distinguishability, winter2005environment, chowdhury2025minimal}. As an illustration, for any quantum channel, it is possible to communicate a bit of classical information, whenever the receiver is LOCC-assisted by the environment. This is a simple corollary of the fact that any two pure orthogonal bipartite quantum states can be distinguished perfectly via LOCC \cite{walgate2002nonlocality}. The impossibility of LOCC discrimination, on the other hand, induces a notion of sharing a secret message among distant parties, decoding of which necessitates some (all) of the parties to join in the same lab \cite{hillery1999secret, terhal2001hiding, eggeling2002hiding, matthews2009distinguishability, rahaman2015quantum}. Similarly, other variants of the phenomenon of quantum indistinguishability under causal constraints \cite{bandyopadhyay2014conclusive, halder2019strong, rout2019genuinely, zhang2019strong, shi2020strong, sen2022local, ghosal2022local} can be cast into several other communication scenarios \cite{frenkel2015classical, dall2017no, patra2022classical}. For instance, consider the notion of \textit{local reducibility}, the goal of which is to eliminate at least one of the possible orthogonal preparations under orthogonality preserving LOCC \cite{halder2019strong}. It can be seen as a close cousin of the multipartite version of the communication task considered in \cite{patra2022classical}. In this setting, the spatially separated receivers have to locally exclude the choice made by the sender. The impossibility of such a task, therefore, introduces a stronger notion of secret sharing, where the information about the message can not even be updated locally by eliminating some of the possibilities.\par

Notably, all the communication settings exploiting quantum indistinguishability,  considered so far, can be characterized in the Holevo-Frenkel-Weiner (HFW) scenario \cite{holevo1973bounds, frenkel2015classical}. The decoding measurement(s) performed by the receiver(s), in these settings, is (are) independent of any further classical input. However, there is another information processing paradigm -- the Wiesner-Ambainis scenario \cite{wiesner1983conjugate, ambainis1999dense, ambainis2002dense} -- for which the decoding measurement also depends upon some classical input given to the receiver. Analogous to such a communication model, which is popularly referred as Random Access Coding (RAC)-like communication set-up, we introduce here a local discrimination task of given orthogonal preparations. The task, in addition, depends on a classical input distributed among the spatially separated receivers \cite{comment}. In accordance with the term \textit{nonlocality}, often referring to the impossibility of local state discrimination \cite{bennett1999quantum}, we coin the term "\textit{conditional nonlocality}" to denote the impossibility of the present task. In the following, we first formally pose the input-dependent local state discrimination task.

\section{The task of Local Random Authentication}
Consider the scenario, where the sender (say Alice) encodes the classical information $k\in\{1,\cdots,N\}$ in one of the $n$-partite orthogonal quantum states $\mathcal{S}:=\{\ket{\psi_k}\}_{k=1}^N$ and distributes among $n$ numbers of spatially separated receivers (say Bob). Subsequently, a referee will distribute (broadcast) one among $N$ inputs $\{\mathcal{Q}_i\}_{i=1}^N$ to Bob(s). The spatially separated receivers then, depending on their input, have to answer a single bit $y\in\{0,1\}$ to the referee. While these inputs may reflect many different queries, we will consider a simple case where every $\mathcal{Q}_i$ represents the query "{\tt whether the state sent by Alice is} $\ket{\psi_i}$ {\tt or not}". On the other hand, the bit $y$ generated by Bob(s) will be $0$ for "{\tt No}" and $1$ for "{\tt Yes}" (see Fig.\ref{f1}). The question asked by the referee is randomly sampled over $\{\mathcal{Q}_i\}_{i=1}^N$. Upon answering the questions successfully, Bob(s) locally authenticate whether the encoded classical index is \(i\) or not, which, in turn, boils down to authenticating whether the shared state is $\ket{\psi_i}$ or not. This justifies the name \textit{Local Random Authentication} (LRA) for the present task. Here it is important to mention that every question \(\mathcal{Q}_i\) corresponds to the task of authenticating the state \(\ket{\psi_i}\) locally. However, akin to the input dependent communication scenario, there is no correlation between the query index \(i\) in \(\mathcal{Q}_i\) and the classical information index \(k\) in the encoded state \(\ket{\psi_k}\), meaning that the additional classical input at the receivers' end bears independent randomness than that of the classical message encoded at the sender's end. In other words, there is no correspondence between the multipartite state shared among the spatially separated receivers and question they have been asked. In the rest of this letter, we will use the terms  \textit{question asked among \(\{\mathcal{Q}_i\}_{i=1}^N\)} and \textit{local authentication of the state among \(\{\ket{\psi_i}\}_{i=1}^N\)} interchangeably. At this juncture, it is important to define a further refinement on the degree of accomplishing the task LRA.
\begin{definition}\label{d1}
    (Complete LRA) A set of multipartite quantum states \(\mathcal{S}:=\{\ket{\psi_k}\}_{k=1}^N\), admits complete LRA, if every question \(\{\mathcal{Q}_i\}_{i=1}^N\) can be perfectly answered under LOCC.
\end{definition}
\begin{definition}\label{d2}
   (Partial LRA) A set of multipartite quantum states \(\mathcal{S}:=\{\ket{\psi_k}\}_{k=1}^N\), admits partial LRA, if at least one question \(\mathcal{Q}_j\), for some \(1\le j \le N\), can be perfectly answered under LOCC.
\end{definition}
It is immediate that complete LRA implies partial LRA. However, the converse does not hold in general. From now on, by LRA we simply refer to complete LRA of a set of quantum states, otherwise mentioned.

Finally we note the premise of LRA is highly aligned with composite hypothesis testing  \cite{berta2021composite, fujiki2025quantum}, under restricted measurement settings. In particular, for the question \(\mathcal{Q}_i\) the null hypothesis: the state is \(\ket{\psi_i}\), is a simple hypothesis, while the alternative hypothesis: the state is not \(\ket{\psi_i}\), renders a wide range of possibilities for the shared state, hence can not be accounted for a simple hypothesis.\par 
\begin{figure}[htb]
    \centering
    \includegraphics[width = 0.8\linewidth]{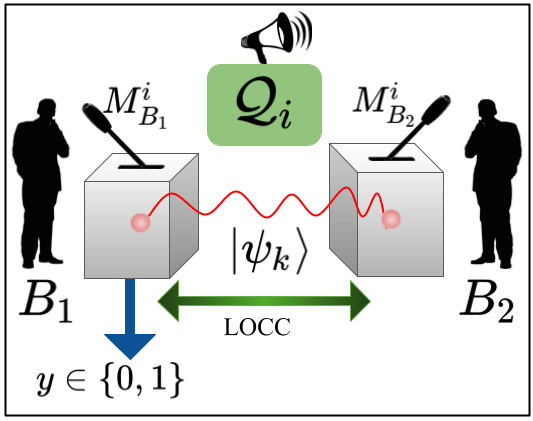}
    \caption{Task of LRA in bipartite settings. A state $\ket{\psi_k}$, randomly chosen from the set of orthogonal quantum states $\mathcal{S}$, is distributed between $B_1$ and $B_2$ and the question $\mathcal{Q}_i$ is asked. Depending upon the question they will execute a LOCC protocol and answer a bit $y\in\{0,1\}$, which stands for "{\tt No}" and "{\tt Yes}" respectively. Since the framework considers unrestricted classical communication between the parties, w.l.g we can assume that the final answer will be given by $B_1$.
    }
    \label{f1}
\end{figure}
To this end, we will move to explore the characteristic features of our proposed task LRA and consequently compare the outcomes from conventional LOCC state discrimination paradigm. The rest of the paper is organized as the following: Section \ref{s3} is dedicated to highlight the main results and comprised of three subsections. In Subsection \ref{ss1}, we compare the strength of LRA with that of the LOCC state discrimination. Thereafter, in Subsection \ref{ss2}, we show that while entangled members are necessary for a set of orthogonal quantum states to exhibit impossibility in the task of LRA, the task itself shows more (conditional) nonlocality with less entanglement phenomenon. Subsection \ref{ss3} mainly characterizes the set of states for which even the partial LRA is impossible. We introduce another weaker variant of LRA, namely the conclusive one in Subsection \ref{ss4}, where we compare its strength with respect to the conclusive LOCC state discrimination and note that the task exhibit more (conditional) nonlocality with less purity. In Section \ref{s4}, we highlight a practical implication of the task LRA in the context of local estimation of the quantum change point problem. Finally, we conclude and discuss promising future directions in Section \ref{s5}.

\par
\section{Main Results}\label{s3}
\subsection{Relationship with LOCC discrimination}\label{ss1}
With the introduction of the task LRA, one may immediately tempted to ask about its strength with respect to the traditional notion of state discrimination. Although intuitive, in the following we will first note that
\begin{theorem}\label{t1}
Perfect LOCC discrimination of a set of quantum states implies their perfect local random authentication, while the converse is not true.
\end{theorem}
\begin{proof}
The fact that perfect LOCC discrimination implies perfect LRA is obvious. Suppose the set $\mathcal{S}:=\{\ket{\psi_i}\}_{i=1}^N$ of orthogonal quantum states can be discriminated under LOCC. Therefore, by identifying the unknown state (say, $\ket{\psi_k}$) given from the set $\mathcal{S}$, the separated parties can answer any of the $N$-possible authentication questions. In particular, their answer will be $0$ for every $\mathcal{Q}_{j\neq k}$ and $1$ for $\mathcal{Q}_k$.

To show that the converse does not hold true, we will consider the following example: 
\begin{eqnarray}\label{e1}
\nonumber\ket{E_1}:= \ket{\phi^{+}}_{B_1B_2}\\
\nonumber\ket{E_2}:= \ket{\phi^-}_{B_1B_2}\\
\ket{E_3}:= \ket{\psi^+}_{B_1B_2}
\end{eqnarray}
It follows from \cite{ghosh2001distinguishability, walgate2002nonlocality} that these three states can not be distinguished perfectly under LOCC. However, they can be authenticated locally by adopting the strategy: $\mathcal{Q}_1\to\sigma_y\otimes\sigma_y,~\mathcal{Q}_2\to\sigma_x\otimes\sigma_x$ and $\mathcal{Q}_3\to\sigma_z\otimes\sigma_z$, where $\sigma_i\otimes\sigma_i$ implies the local measurement of the Pauli observable $\sigma_i$ for both the parties. In particular, the spatially separated parties will answer $y=0$ whenever their outcomes are correlated and $y=1$ otherwise. 
\end{proof}

\subsection{Necessity of entanglement}\label{ss2}
All the nonlocal aspects in terms of local state discrimination (known till now) also hold for pure product states \cite{bennett1999quantum, rinaldis2004distinguishability, niset2006multipartite, feng2009characterizing, zhang2014nonlocality, xu2017local, halder2018several, halder2019strong, rout2021multiparty}, often termed as "\textit{Nonlocality Without Entanglement}". Taken in isolation, one can even locally activate such nonlocal features starting from  locally distinguishable product preparations \cite{li2022genuine, ghosh2022activating}. Therefore, such nonlocal features can not certify the presence of entanglement in the preparation device. However, in the following, we will first show that LRA can be accomplished perfectly for any product state, and as a consequence, we will conclude that any \textit{conditional nonlocal} set must contain entanglement.
\begin{theorem}\label{t2}
Any pure product state chosen from a set of orthogonal quantum states can be authenticated locally.
\end{theorem}
\begin{proof}
Let us consider a set of orthogonal quantum states, possibly mixed, $\mathcal{S}:=\{\rho_i\}_{i=1}^n$ shared between $N$-spatially separated parties. Moreover, the set $\mathcal{S}$ contains atleast one pure product state $\ket{\psi_j}=\bigotimes_{k=1}^N\ket{\xi_k^{(j)}}$. Note that, since all the $\rho_i\in\mathcal{S}$ are mutually orthogonal, $\rho_j=|\psi_j\rangle\langle \psi_j|\perp\rho_i,~\forall i\in\{1,2,\cdots,n\}\setminus\{j\}$.

Now, for the question $\mathcal{Q}_j$, i.e., whether the shared state is $\ket{\psi_j}$ or not, each of the parties can simply perform the measurements $M_{B_k}^{(j)}:=\{\ketbra{\xi_k^{(j)}}{\xi_k^{(j)}},\mathbb{I}-\ketbra{\xi_k^{(j)}}{\xi_k^{(j)}}\},~\forall k\in\{1,2,\cdots,N\}$ on their local constituents.
Accordingly, the projector $\mathcal{P}_j:=\bigotimes_{k=1}^{N}\ketbra{\xi_k^{(j)}}{\xi_k^{(j)}}$ of the measurement $\bigotimes_{k=1}^{N} M_{B_k}^{(j)}$ clicks only when they share the state $\ket{\psi_j}$. Conversely, no other effects click for this state. Hence, they will answer $1$ when $\mathcal{P}_j$ clicks and $0$ otherwise. 
\end{proof}
From Theorem \ref{t2}, one can easily derive the following corollary. 
\begin{corollary}\label{c1}
Any set of orthogonal pure product states admits complete local authentication.
\end{corollary}
After establishing the strict necessity of entanglement to exhibit the impossibility of LRA, one may be tempted to ask whether the complexity associated with LRA task scales with the number of entangled states present in the given ensemble. Without going to a rigorous quantification of the complexity, we will now show that such a monotonic relation does not hold true in general. In particular, we will identify a set $\mathcal{P}$, such that any $\mathcal{S}$ containing $\mathcal{P}$, can not be authenticated locally if the elements of the set $\mathcal{S}\setminus\mathcal{P}$ are product states. However, there exists a set $\mathcal{S}^\prime\supset\mathcal{P}$, which, interestingly,  allows complete LRA, if a particular entangled state is present in $\mathcal{S}^\prime\setminus\mathcal{P}$. Mimicking the phrase "\textit{more nonlocality with less entanglement}", used in \cite{horodecki2003local}, we put the following proposition.
\begin{proposition}\label{p1}
The impossibility of LRA exhibits more (conditional) nonlocality with less entanglement phenomenon.
\end{proposition}
\begin{proof}
    Here we will consider a minimalistic  scenario:  $\mathcal{S}$ is a collection of \textit{three} two-qubit states, along with $\mathcal{P}$ which contains exactly two maximally entangled states $\{\ket{\phi^+},\ket{\phi^-}\}$, shared between $B_1$ and $B_2$. Now, the only possible product state which is orthogonal to both of them is either $\ket{01}$ or, $\ket{10}$. Without loss of any generality, we can assume $\mathcal{S}:= \mathcal{P}\cup\{\ket{01}_{B_1B_2}\}=\{\ket{\phi^+}_{B_1B_2},\ket{\phi^-}_{B_1B_2},\ket{01}_{B_1B_2}\}$, since there always exist a local unitary, which interchanges between $\ket{01}$ and $\ket{10}$, keeping $\ket{\phi^+}$ and $\ket{\phi^-}$ invariant.

    According to Theorem \ref{t2}, $B_1$ and $B_2$ can authenticate the state $\ket{01}$ locally, i.e., they can answer the question $\mathcal{Q}_3$ perfectly, just by performing $\sigma_z\otimes\sigma_z$ measurement on their local constituents. Now, suppose there exists a locally implementable POVM, by performing which they can answer the question $\mathcal{Q}_1$ perfectly. With no loss of generality, we assume that $B_1$ starts the protocol by performing a POVM $\{M_m^{(1)}\}_m$, where $M_m^{(1)}$ is a $2\times2$ positive semi-definite matrix such that $\sum_m M_m^{(1)}=\mathbb{I}$. If $B_1$ obtains the result $m$ then the protocol can be carried out further, only if the post-measurement state of $\ket{\phi^+}$ remains orthogonal to the post-measurement states of $\ket{\phi^-}$ and $\ket{01}$. Specifically, we write 
    \begin{eqnarray}\label{e2}
     \nonumber\bra{\phi^-}(M_m^{(1)})_{B_1}\otimes\mathbb{I}_{B_2}\ket{\phi^+}=0,\\
        \text{ and }\bra{01}(M_m^{(1)})_{B_1}\otimes\mathbb{I}_{B_2}\ket{\phi^+}=0   
    \end{eqnarray}
    Now considering
    \begin{align}
    \nonumber
	M_m^{(1)}:= 
	\begin{pmatrix}
		\alpha & \beta \\
		\beta^* & \gamma 
	\end{pmatrix},
\end{align}
and using Eq. (\ref{e2}), we obtain $\alpha=\gamma$ and $\beta=0$, i.e., $M_m^{(1)}=\alpha \mathbb{I}$. This rules out the possibilities of any nontrivial POVM on $B_1$'s side, by which they can authenticate the state $\ket{\phi^+}_{B_1B_2}$. Therefore, the set $\mathcal{S}$ does not allow complete local authentication.

On the other hand, if we consider $\mathcal{S}^\prime:=\mathcal{P}\cup\{\ket{\psi^+}_{B_1B_2}\}$, then $B_1$ and $B_2$ can perfectly authenticate the set $\mathcal{S}^\prime$ locally (as follows from the example given in the proof of Theorem \ref{t1}).
\end{proof}

   Note that the conditions in Eq. (\ref{e2}) are weaker than that of \textit{orthogonality preserving local measurement}, introduced in \cite{walgate2002nonlocality}, which is an important tool to establish the \textit{irreducibility} \cite{halder2019strong} for a given set of orthogonal states.
   
   It is also important to mention that the phenomenon of more (conditional) nonlocality with less entanglement for the task LRA is more surprising than the similar results for LOCC discrimination \cite{horodecki2003local}. This is because while none of the product states possess \textit{conditional nonlocaltiy}, there are several examples of \textit{nonlocal} product states in the question of LOCC discrimination \cite{halder2018several}. 

Although the set $\mathcal{S}:=\mathcal{P}\cup\{\ket{01}\}$ in Proposition \ref{p1} does not allow the perfect accomplishment of LRA, the question $\mathcal{Q}_3$ can always be answered correctly. This gives a concrete example of partial LRA for a given set of orthogonal quantum states. It is obvious from Theorem \ref{t2} that if a set of quantum states does not allow even partial LRA, then the set necessarily contains no product state. In the following, we will further emphasize on the structure of such strictly (conditional) nonlocal sets of multipartite quantum states.
\subsection{Impossibility of Local Random Authentication}\label{ss3}
In the multipartite settings, quantum entanglement exhibits rich topological structure, giving rise to infinitely many inequivalent classes under LOCC transformations \cite{sauerwein2018transformations}. Depending on the entanglement structure, authentication tasks can be formulated under various causal constraints—for instance, allowing arbitrary subsets of parties to perform joint decoding measurements, while others remain spatially separated. In contrast, the setting of LRA assumes complete spatial separation: no two receivers may share a laboratory or perform joint operations. Under this strict constraint, we establish a no-go theorem ruling out the possibility of even partial LRA, as formalized in the following.

\begin{lemma}\label{l1}
    If a set of orthogonal states can not be discriminated conclusively under LOCC, even with an arbitrarily small probability, then the set does not even allow partial LRA. 
\end{lemma}
\begin{proof}
     We will prove the result conversely. Consider a set of orthogonal quantum states $\mathcal{S}:=\{\ket{\psi_i}\}_{i=1}^N$ which allows the accomplishment of partial LRA. This means, there exists at least one $\ket{\psi_k}\in\mathcal{S}$ which can be authenticated locally, i.e., the question $\mathcal{Q}_k$ can be answered correctly, using a two-outcome LOCC implementable quantum measurement $\mathcal{M}_{k}:=\{M_k^{(0)},M_k^{(1)}\}$.

     Now, consider a situation of LOCC state discrimination, where an arbitrary state from $\mathcal{S}$ is distributed among spatially separated parties. It is obvious from above that they can at least perform the LOCC implementable measurement $\mathcal{M}_{k}$ and can conclusively identify the unknown state. In particular, the state will be $\ket{\psi_k}$ whenever $M_k^{(1)}$ clicks, otherwise inconclusive. This implies, the set $\mathcal{S}$ is conclusively LOCC distinguishable with a probability no less than $\frac{1}{N}$.
\end{proof}
\noindent
 Numerically the probability exactly matches with that of the random guessing, however, unlike the random guessing, the conclusive state discrimination never makes a wrong identification \cite{chefles1998unambiguous}.
 
Using the above lemma we will now conclude about the possibilities of LRA for a complete orthonormal basis of multipartite quantum states. 

\begin{theorem}\label{t3}
    Consider a complete set of orthogonal basis states $\mathcal{S}_{n}:=\{\ket{\psi_i}_{B_1,B_2,\cdots,B_n}\}\in\bigotimes_{k=1}^n\mathbb{C}^{d_k}$. If none of these states is fully product, then the set does not even allow partial LRA.
\end{theorem}
\begin{proof}
     Consider a situation with
    $2n$ numbers of parties $\{A_1,\cdots,A_n,B_1,\cdots,B_n\}$, where each pair of parties $A_k$ and $B_k$ share a maximally entangled state $\ket{\phi^+_k}=\frac{1}{\sqrt{d_k}}\sum_{l=0}^{d_k-1}\ket{l}_{A_k}\ket{l}_{B_k},$ $\forall k\in\{1,2,\cdots,n\}$. 
    Evidently, the joint state will take the form,
    \begin{equation}\label{e3}
        \ket{\Phi^+}_{A_1B_1\cdots A_nB_n}=\frac{1}{\sqrt{d_1d_2\cdots d_n}}\bigotimes_{k=1}^n(\sum_{l=0}^{d_k-1}\ket{l}_{A_k}\ket{l}_{B_k}).   
    \end{equation}
It is easy to show that the state in Eq.(\ref{e3}) can be effectively treated as a maximally entangled state 
        \begin{equation*}
        \frac{1}{\sqrt{d_1d_2\cdots d_n}}\sum_{l=0}^{d_1d_2\cdots d_n-1}\ket{l}_{A_1A_2\cdots A_n}\ket{l}_{B_1B_2\cdots B_n}
        \end{equation*}
        for $(\bigotimes_{k=1}^{n}\mathbb{C}^{d_k})^{\otimes 2}$ in the $\mathcal{A}:=\{A_1A_2\cdots A_n\}$ vs. $\mathcal{B}:=\{B_1B_2\cdots B_n\}$ bipartition. Finally, using the $\mathbf{U}\otimes\mathbf{U}^*$ invariance of $\ket{\Phi^+}_{\mathcal{A}, \mathcal{B}}$, the state can be rewritten as
        \begin{equation}\label{e4}
            \ket{\Phi^+}_{\mathcal{A},\mathcal{B}}=\frac{1}{\sqrt{d_1d_2\cdots d_n}}\sum_{l=0}^{d_1d_2\cdots d_n-1}\ket{\xi_l}_{\mathcal{A}}\ket{\xi_l^*}_{\mathcal{B
            }},
        \end{equation}
      where $\{\ket{\xi_l}\}_{l=0}^{d_1d_2\cdots d_n-1}$ is any arbitrary orthonormal basis in $\bigotimes_{k=1}^{n}\mathbb{C}^{d_k}$ and $^*$ denotes the complex conjugation. Evidently, the orthogonal basis set $\{\ket{\psi}_{B_1B_2\cdots B_n}\}$ can also be identified as $\{\ket{\xi_l}_{\mathcal{B}}\}$.\par
    Therefore, if there is a $\ket{\psi_m}\in\mathcal{S}_n$ which is not fully separable, i.e., entangled in some $\{B_1\cdots B_k\}$ vs. $\{B_{k+1}\cdots B_n\}$ bipartition, then by identifying it conclusively under LOCC they can inform  $\{A_1,\cdots,A_n\}$, who are sharing the state $\ket{\psi_m^*}$, entangled in the $\{A_1\cdots A_k\}$ vs $\{A_{k+1}\cdots A_n\}$ bipartition.   However, the state in Eq. (\ref{e3}), and hence Eq. (\ref{e4}), is separable by construction in any possible bipartition $\{A_1B_1\cdots A_mB_m\}$ vs. $\{A_{m+1}B_{m+1}\cdots A_nB_n\}$. Since entanglement in any bipartition can not be increased under LOCC, $\ket{\psi_m}$ must be fully product. In other words, the conclusive local discrimination of $\mathcal{S}_n$ necessitates the presence of at least one product state.\par
    Finally, with the help of Lemma \ref{l1}, we conclude that none of the elements of the set $\mathcal{S}_n$ can be authenticated locally if it contains no product states. This completes the proof.
\end{proof}
\subsection{Conclusive LRA \textit{vs.} Conclusive LOCC}\label{ss4}
In Lemma \ref{l1}, we have shown that the complexity associated with conclusive state discrimination, under LOCC, is not higher than that of partial LRA. A natural question arises at this point whether the converse holds true or not. In the following, we will answer it in negative by providing an example even in the simplest bipartite scenario.  
\begin{proposition}\label{p2}
    The impossibility of partial LRA does not imply the impossibility of conclusive local discrimination.  
\end{proposition}
\begin{proof}
    Consider a set $\mathcal{S}=\{\ket{\psi_k}\}_{k=1}^8$ of bipartite entangled states in $\mathbb{C}^3\otimes\mathbb{C}^3$ shared between $B_1$ and $B_2$:
    \begin{eqnarray*}
        \ket{\psi_k}_{B_1B_2}&=&\ket{E_k}_{B_1B_2}, \forall k\in\{1,2,3\}\\
        \ket{\psi_4}_{B_1B_2}&=&\frac{1}{\sqrt{3}}(\sqrt{2}\ket{\psi^-}+\ket{22})_{B_1B_2},\\
        \ket{\psi_5}_{B_1B_2}&=&\frac{1}{\sqrt{2}}(\ket{12}+\ket{21})_{B_1B_2},\\
       \ket{\psi_6}_{B_1B_2}&=&\frac{1}{\sqrt{2}}(\ket{12}-\ket{21})_{B_1B_2},\\
        \ket{\psi_7}_{B_1B_2}&=&\frac{1}{\sqrt{2}}(\ket{02}+\ket{20})_{B_1B_2},\\
        \ket{\psi_8}_{B_1B_2}&=&\frac{1}{\sqrt{2}}(\ket{02}-\ket{20})_{B_1B_2},
    \end{eqnarray*}
    where, $\{\ket{E_k}\}_{k=1}^3$ are three of the two-qubit Bell states, as in Eq. (1).
    
    First note that the set of states can be conclusively discriminated under LOCC with a minimum probability of $\frac{1}{24}$. The protocol is simply to perform the measurement $\mathcal{M}:=\{\ketbra{2}{2},\mathbb{I}_3-\ketbra{2}{2}\}$ by both of the parties and they can conclude the given state to be $\ket{\psi_4}$ whenever the projector $\ketbra{2}{2}_{B_1}\otimes\ketbra{2}{2}_{B_2}$ clicks.\par
    
    Next, we will show that none of the states $\{\ket{\psi_k}\}_{k=1}^8$ can be authenticated by LOCC. Since, all the members of $\mathcal{S}\setminus \{\ket{\psi_4}\}$ are connected via local unitary, it is sufficient to show the impossibility of LOCC authentication of only one $\ket{\psi_i} \in \mathcal{S}\setminus \{\ket{\psi_4}\}$. We chose $\ket{\psi_1}=\ket{\phi^+_{B_1B_2}}$ for this purpose. Moreover, since the set $\mathcal{S}$ is party symmetric, with out loss of generality, one can begin their analysis with Bob1. \par

    To locally authenticate $\ket{\phi^+_{B_1B_2}}$ among $\mathcal{S}$, there must exist a non-trivial POVM on Bob1's side which preserves the orthogonality between $\ket{\phi^+_{B_1B_2}}$ and $\mathcal{S}\setminus \{\ket{\phi^+}_{B_1B_2}\}$. In other words, Bob1 must be able to perform a measurement $\{F_n^{(1)}\}_n$, where $F_n^{(1)}$ is a $3\times3$ positive semi-definite matrix such that $\sum_n F_n^{(1)}=\mathbb{I}$.
    We will take into account the following four equations:
     \begin{eqnarray}\label{e2}
     \bra{\phi^-}(F_n^{(1)})_{B_1}\otimes\mathbb{I}_{B_2}\ket{\phi^+}&=&0\label{e21}\\
     \bra{\psi^+}(F_n^{(1)})_{B_1}\otimes\mathbb{I}_{B_2}\ket{\phi^+}&=&0\label{e22} \\
     (\bra{\psi^-}+\bra{22})(F_n^{(1)})_{B_1}\otimes\mathbb{I}_{B_2}\ket{\phi^+}&=&0\label{e23} \\
     (\bra{12}+\bra{21})(F_n^{(1)})_{B_1}\otimes\mathbb{I}_{B_2}\ket{\phi^+}&=&0\label{e24} \\
     (\bra{02}+\bra{20})(F_n^{(1)})_{B_1}\otimes\mathbb{I}_{B_2}\ket{\phi^+}&=&0\label{e25}
    \end{eqnarray}

    Consider.
    \begin{align*}
        F_n^{(1)}:= 
	\begin{pmatrix}
		r_1 & \alpha & \beta \\
		\alpha^* & r_2 & \gamma \\
            \beta^* & \gamma^* & r_3 
	\end{pmatrix},
\end{align*}
Starting with Eq (\ref{e21}), we determine that $r_1 =r_2$. Subsequently, examining Eqs.(\ref{e22}) and (\ref{e23}), we get $\alpha= 0$. Furthermore, Eqs.(\ref{e24}) and (\ref{e25}) lead to the conclusions that $\beta,\gamma = 0$ respectively. It is now clear from the structure of $F_n^{(1)}$, that the only nontrivial orthogonality preserving measurement is $\ketbra{2}{2}$ \textit{vs}. $(\mathbb{I}-\ketbra{2}{2})$. Upon, performing this measurement, if the effect $\ketbra{2}{2}$ clicks, they can confirm that the given state is not $\ket{\phi^+}$.\par
However, if the other effect clicks, the set $\mathcal{S}$ will be updated accordingly to $\mathcal{\Tilde{S}}$. Notice, the first four states of the updated set $\mathcal{\Tilde{S}}$ will be the four Bell states, while both $\ket{\psi_5}$ and $\ket{\psi_6}$ ($\ket{\psi_7}$ and $\ket{\psi_8}$) will collapse to $\ket{12}_{B_1B_2}$ ($\ket{02}_{B_1B_2}$). On the set $\mathcal{\Tilde{S}}$ of updated states Bob2 can perform another POVM $\{F_n^{(2)}\}_n$ with $F_n^{(2)}\geq0,~\forall n$ and $\sum_n F_n^{(2)}=\mathbb{I}$. Similar to the previous case, the same orthogonality conditions will lead to the following set of equations:
\begin{eqnarray}
\nonumber\bra{\psi_k}\mathbb{I}_{B_1}\otimes(F_n^{(2)})_{B_2}\ket{\phi^+}&=&0,~k\in\{2,3\}\\
\nonumber\bra{\psi^-}\mathbb{I}_{B_1}\otimes(F_n^{(2)})_{B_2}\ket{\phi^+}&=&0\\
     \nonumber\bra{12}\mathbb{I}_{B_1}\otimes(F_n^{(2)})_{B_2}\ket{\phi^+}&=&0\\
     \nonumber\bra{02}\mathbb{I}_{B_1}\otimes(F_n^{(2)})_{B_2}\ket{\phi^+}&=&0
    \end{eqnarray}
    A simple numerical exercise from the above conditions establishes that the only possible measurement in Bob2's side is $\{\ketbra{2}{2},\mathbb{I}-\ketbra{2}{2}\}$. Performing this measurement if $\ketbra{2}{2}$ clicks they can certify that the state is not $\ket{\phi^+}_{B_1B_2}$. However, for the effect $\mathbb{I}-\ketbra{2}{2}$ the updated set of states will be $\{\ket{E_k}_{B_1B_2}\}_{k=1}^4$, i.e., the four Bell states in $\mathbb{C}^2\otimes\mathbb{C}^2$. 
Now, invoking Theorem 3 of the main text, one can easily argue that it is impossible to further authenticate the state $\ket{\phi^+}$.\par
To check the local authenticity of the state $\ket{\psi_4}_{B_1B_2}$, again we proceed by analysing orthogonality preserving local POVM. Consider a general POVM on Bob1's lab, 
\begin{align}
    \nonumber
	\mathcal{E}_n^{(1)}:= 
	\begin{pmatrix}
		a_1 & b & c \\
		b^* & a_2 & d \\
            c^* & d^* & a_3 
	\end{pmatrix}
\end{align}
For $\mathcal{E}_n^{(1)}$ to be a potential measurement at Bob1's side which can authenticate the state $\ket{\psi_4}_{B_{1}B_{2}}$ locally, the following set of equations must be satisfied:
$$ \bra{\psi_k} (\mathcal{E}_n^{(1)})_{B_1}\otimes\mathbb{I}_{B_2} \ket{\psi_4} = 0, ~\forall k \in \{1,2,3\}\cup\{5,6\}.$$

By solving these equations, we get $a_1=a_2$ and $b,c,d=0$. These leave Bob1 with the only choice of measurement $\{\ketbra{2}{2},\mathbb{I}-\ketbra{2}{2}\}$. Now, if the effect $(\mathbb{I}-\ketbra{2}{2})$ clicks the set of states gets updated to that of $\Tilde{S}$. Hence, following the same line of argument as above, one can conclude that it is impossible to authenticate the state $\ket{\psi_4}$. This proves our claim.
\end{proof}

Note that if both the parties perform the computational basis measurement $\{\ketbra{0}{0},\ketbra{1}{1},\ketbra{2}{2}\}$ in their local lab, they can conclusively answer `{\tt No}' for any of the questions \(\{\mathcal{Q}_i\}_{i=1}^8\). For instance, both for \(\mathcal{Q}_1\) and \(\mathcal{Q}_2\) they can surely answer `{\tt No}' except when the projectors \(\ketbra{0}{0}\otimes\ketbra{0}{0}\) or \(\ketbra{1}{1}\otimes\ketbra{1}{1}\) clicks. Similarly, one can argue the same for all the other cases. However, interesting scenario arises for the question \(\mathcal{Q}_4\). The parties can also conclusively answer the question with `{\tt Yes}' when both of them clicks \(\ketbra{2}{2}\) projector. Therefore, although the set $\mathcal{S}$ does not permit even partial LRA; any of the questions \(\{\mathcal{Q}_i\}_{i=1}^8\) can be answered conclusively. Moreover, for the question \(\mathcal{Q}_4\), both possibilities, i.e, `{\tt Yes}', as well as `{\tt No}' can  be answered conclusively. This readily instigates the notion of conclusive LRA for a set of multipartite quantum states which we formally define in the following. 
\begin{definition}\label{d3}
(Conclusive LRA) A set of orthogonal multipartite quantum states is said to admit conclusive LRA, if every questions $\{\mathcal{Q}_k\}_k$ can be conclusively answered with a nonzero probability. 
\end{definition}
Notably, Lemma \ref{l1} and Proposition \ref{p2}, hints that the impossibility of partial LRA is a stronger aspect of nonlocality than that of the conclusive LOCC discrimination. However, in question of conclusive LRA, a relation similar to that of Theorem \ref{t1} can be obtained. Precisely,
\begin{theorem}\label{t4}
    Conclusive LOCC distinguishability of a set of quantum states implies the possibility of conclusive LRA, while the converse does not hold true.
\end{theorem}
\begin{proof}
     For any set of pure orthogonal quantum states the possibility of conclusive LRA can be argued trivially. In fact, every set of orthogonal pure quantum states admit conclusive LRA. This is because, for every quantum state $\ket{\psi_{B_1B_2\cdots B_n}}$, there always exists a pure product state $\ket{\Phi_{B_1B_2\cdots B_n}}:=\bigotimes_{k=1}^n\ket{\phi_{B_k}}$, such that 
 $$\langle\psi_{B_1B_2\cdots B_n}|\Phi_{B_1B_2\cdots B_n}\rangle=0.$$
 Therefore, for any set of orthogonal states $\mathcal{S}\ni\ket{\psi}$, whenever the question $\mathcal{Q}_{\psi}$ asked, the spatially separated parties can perform the measurement $\{\ketbra{\phi_{B_k}}{\phi_{B_k}},\mathbb{I}-\ketbra{\phi_{B_k}}{\phi_{B_k}}\}_{k=1}^n$ in each of their individual lab. If $\ketbra{\phi_{B_k}}{\phi_{B_k}}$ clicks in all the labs, then they can conclusively answer the question $\mathcal{Q}_{\psi}$, i.e., the state is not $\ket{\psi_{B_1B_2\cdots B_n}}$. 

  On the other hand, consider any set of complete orthogonal basis, containing no product states. It is clear from the proof of Theorem \ref{t3} that none of them can be conclusively discriminated under LOCC.
\end{proof}

While from the proof of Theorem \ref{t4}, it seems that the notion of conclusive LRA can be accomplished for any set of orthogonal pure quantum states, in the following we will show that the same notion can not be trivially extended for mixed states. In fact, such an impossibility gives rise to another surprising feature for the conditional nonlocality which is depicted in the following proposition.
\begin{proposition}\label{p3}
    The task of conclusive LRA depicts more (conditional) nonlocality with less purity.
\end{proposition}
    \begin{proof}
    Consider the following class of pair of orthogonal mixed states in \(\mathbb{C}^4\otimes\mathbb{C}^4\):
    \begin{align*}
        \rho_1(\theta)&=\frac{1}{8}\sum_{i=1}^8\ketbra{\psi_i^+}{\psi_i^+}_{B_1B_2},\\\text{and }\rho_2(\theta)&=\frac{1}{8}\sum_{i=1}^8\ketbra{\psi_i^-}{\psi_i^-}_{B_1B_2}
    \end{align*}
    where, 
    \begin{align*}
        \ket{\psi_1^{\pm}}_{B_1B_2}&=\frac 1{\sqrt{2}}(\ket{00}\pm\ket{11})_{B_1B_2}\\\ket{\psi_2^{\pm}}_{B_1B_2}&=\frac 1{\sqrt{2}}(\ket{22}\pm\ket{33})_{B_1B_2}\\\ket{\psi_3^{\pm}}_{B_1B_2}&=\frac 1{\sqrt{2}}(\ket{20}\mp\ket{02})_{B_1B_2}\\\ket{\psi_4^{\pm}}_{B_1B_2}&=\frac 1{\sqrt{2}}(\ket{30}\pm\ket{03})_{B_1B_2}\\\ket{\psi_5^{\pm}}_{B_1B_2}&=\frac 1{\sqrt{2}}(\ket{13}\pm\ket{31})_{B_1B_2}\\\ket{\psi_6^{\pm}}_{B_1B_2}&=\frac 1{\sqrt{2}}(\cos \theta \ket{01} +\sin\theta \ket{23}\mp \ket{12})_{B_1B_2}\\\ket{\psi_7^{\pm}}_{B_1B_2}&=\frac 1{\sqrt{2}}(\cos \theta \ket{10} +\sin\theta \ket{32}\mp \ket{21})_{B_1B_2}\\\ket{\psi_8^{\pm}}_{B_1B_2}&=\frac 1{\sqrt{2}}[(\sin \theta \ket{01} -\cos\theta \ket{23}) \pm\\& (\sin \theta \ket{10} -\cos\theta \ket{32})]_{B_1B_2}
    \end{align*}
    and \(0<\theta<\pi/2\). Now, it is easy to see that \(\rho_1(\theta)\perp\rho_2(\theta)\) for all \(0<\theta<\pi/2\). In the following we will show that for each \(\theta\in(0,\pi/2)\), one can not authenticate the set \(\{\rho_1(\theta),\rho_2(\theta)\}\), even conclusively.

   Since, for each such set, there are only two possible quantum states that can be shared between Bob1 and Bob2, the task conclusive LRA also implies their conclusive LOCC discrimination. We will prove the proposition by contradiction. Suppose, without loss of any generality, the question \(Q_1\), i.e., whether the shared state is \(\rho_1(\theta)\) or not,  can be answered conclusively under LOCC. Since, the set of all LOCC is a strict subset of the set of separable measurements \cite{chitambar2014everything}, we can safely consider a separable measurment of the form 
   \[\{P_k\}_k \text{ with, }P_k=\sum_{i_k}A_{i_k}\otimes B_{i_k}, \forall k~\&~\sum_k P_k=\mathbb{I}_{B_1B_2},\]
   where, both \(A_{i_k}, B_{i_k}\geq 0,~\forall i_k\), which successfully performs the task of conclusive authentication. Further assume when the projector \(P_c=\sum_{k\in k_c}P_k\) clicks, then the spatially separated parties can conclusively answer the question \(\mathcal{Q}_1\) as {\tt "No"} (or, they answer the question \(\mathcal{Q}_2\) as {\tt "Yes"} conclusively). Therefore, \(\forall k\in k_c\text{ and }\forall i_k\), we have
   \begin{align*}
       &\Tr(P_c~\rho_1(\theta))=0\\\Rightarrow &\Tr(P_k~\rho_1(\theta))=\Tr[(A_{i_k}\otimes B_{i_k})~\rho_1(\theta)]=0.
   \end{align*}
   Now, each of the \(A_{i_k}\) and \(B_{i_k}\) being positive semi-definite and hence normal, we can also write
   \[\Tr[(\ketbra{\phi_{i_k}^l}{\phi_{i_k}^l}\otimes \ketbra{\xi_{i_k}^m}{\xi_{i_k}^m})\rho_1(\theta)]=0,\]
   where, \(\ket{\phi_{i_k}^l}\) and \(\ket{\xi_{i_k}^m}\) are respectively the \(l^{th}\) and \(m^{th}\) eigenvectors of the operators \(A_{i_k}\) and \(B_{i_k}\). Therefore, \(\ket{\phi_{i_k}^l}\otimes\ket{\xi_{i_k}^m}\in\operatorname{ Kernel}(\rho_1(\theta))\). Since \(\rho_1(\theta)\perp\rho_2(\theta)\), we have, 
   \[\ket{\phi_{i_k}^l}\otimes\ket{\xi_{i_k}^m}\in\operatorname{ Span}(\rho_2(\theta)).\]
   But, for every \(\theta\in(0,\pi/2)\), the span of the state \(\rho_2(\theta)\) contains no product states \cite{duan2009super}.

  Similarly, one can argue that the span of \(\rho_1(\theta)\) also contains no product states, which leads us to conclude that the set \(\{\rho_1(\theta),\rho_2(\theta)\}\), for all \(\theta\in(0,\pi/2)\), do not admit even conclusive LRA. On the contrary, it is clear from the proof of Theorem \ref{t4}, that any set of orthogonal pure quantum states admits conclusive LRA, that is, for every question \(\mathcal{Q}_k\) asked to the spatially separated parties, they can atleast answer the question "{\tt No}" conclusively. 
  \end{proof}
At this juncture, it is important to note few points. From the proof of Proposition \ref{p3}, one might think that impossibility of conclusive LRA, in case of two orthogonal mixed states, boils down to impossibility of conclusive LOCC discrimination. In other words, for two mixed states, the phenomena of more conditional nonlocality with less purity and more nonlocality with less purity \cite{bandyopadhyay2011more} represent same aspect of nonlocality. However this is not the case. Indeed, the notion of more conditional nonlocality with less purity captures stronger notion of nonlocality. While it is obvious that, even for two mixed states, all the sets which does not admit conclusive LRA, also does not admit conclusive LOCC discrimination, the vice versa is not true. For example take note of the paradigmatic example of the states \(\sigma_1,\sigma_2\in\mathbb{C}^{3}\otimes\mathbb{C}^{3}\), where:
    \begin{align*}
        \sigma_1=\frac{1}{5}\sum_i|\phi_i\rangle_{B_1B_2}\langle\phi_i|,\quad \sigma_2= \frac{1}{4}\left[\mathbb{I}_9-\sum_i|\phi_i\rangle_{B_1B_2}\langle\phi_i|\right]
    \end{align*}
    The states \(\{\ket{\phi_i}_{B_1B_2}\}_i\) forms the unextendible product basis in \(\mathbb{C}^3\otimes\mathbb{C}^3\) and are given as:
    \begin{align*}
        \ket{\phi_1}_{B_1B_2} &= \frac{1}{\sqrt{2}}(\ket{0}\ket{0-1})_{B_1B_2},\\
        \ket{\phi_2}_{B_1B_2} &= \frac{1}{\sqrt{2}}(\ket{0-1}\ket{2})_{B_1B_2},\\        
        \ket{\phi_3}_{B_1B_2} &= \frac{1}{\sqrt{2}}(\ket{2}_{B_1}\ket{1-2})_{B_1B_2},\\
        \ket{\phi_2}_{B_1B_2} &= \frac{1}{\sqrt{2}}(\ket{1-2}\ket{0})_{B_1B_2},\\
        \ket{\phi_2}_{B_1B_2} &= \frac{1}{3}(\ket{0+1+2}\ket{0+1+2})_{B_1B_2}.\\
    \end{align*}
    It is easy to see that the set \(\{\sigma_1,\sigma_2\}\) does not allow conclusive LOCC discrimination \cite{bandyopadhyay2011more}. However, the set admits conclusive LRA. To perform this, Bob1, without loss of generality, can perform \(\{\ketbra{0_{B_1}}{0_{B_1}},\mathbb{I}_3-\ketbra{0_{B_1}}{0_{B_1}}\}\) measurement and Bob2 can perform \(\{1/2\ketbra{(0-1)_{B_2}}{(0-1)_{B_2}},\mathbb{I}_3-1/2\ketbra{(0-1)_{B_2}}{(0-1)_{B_2}}\}\) measurement. If the projectors corresponding to \(\ketbra{0_{B_1}}{0_{B_1}}\otimes1/2\ketbra{(0-1)_{B_2}}{(0-1)_{B_2}}\) clicks, then they will answer {\tt "No"} if the question is \(\mathcal{Q}_2\) and {\tt "Yes"} if the question is \(\mathcal{Q}_1\).  

For the convenience of the readers we summarize all the important implications between the LRA task and the traditional LOCC state discrimination task in Fig.\ref{fig:lra-locc}.
\begin{figure}[ht]
  \centering
  \resizebox{0.5\textwidth}{!}{%
  \begin{tikzpicture}[
    >=Stealth,              
    node distance=10em and 15em,
    every node/.style={font=\large}
  ]

    \node (CLRA)  {Complete LRA};
    \node (PLOCC) [right=of CLRA] {Perfect LOCCD};
    \node (PLRA)  [below=of CLRA] {Partial LRA};
    \node (CLOCC) [below=of PLOCC] {Conclusive LOCCD};
    \node (CLRA2) [above left=1.5cm and 1cm of CLOCC ] {Conclusive LRA};

    \draw[->,transform canvas={shift={(-6pt,0pt)}}] (CLRA.south) to node[pos=0.5,sloped,below]{Proposition\,\ref{p1}} (PLRA.north);
    \draw[->,transform canvas={shift={(6pt,0pt)}}] (PLRA.north) -- node[midway,pos=0.6]{\color{red}$\times$} (CLRA.south);

    \draw[->,transform canvas={shift={(0pt,4pt)}}] (PLOCC.west) -- node[above]{Theorem\,\ref{t1}} (CLRA.east);
    \draw[->,transform canvas={shift={(0pt,-4pt)}}] (CLRA.east) -- node[midway]{\color{red}$\times$} (PLOCC.west);

    \draw[->,dashed,transform canvas={shift={(6pt,0pt)}}] (PLOCC.south) -- (CLOCC.north);
    \draw[->,dashed,transform canvas={shift={(-6pt,0pt)}}] (CLOCC.north) -- node[midway]{\color{red}$\times$} (PLOCC.south);

    \draw[->,transform canvas={shift={(1pt,4pt)}}] (PLRA.east) -- node[above]{Lemma\,\ref{l1}} (CLOCC.west);
    \draw[->,transform canvas={shift={(1pt,-4pt)}}] (CLOCC.west) -- node[midway]{{\large \color{red}$\times$}} node[below]{Proposition\,\ref{p2}} (PLRA.east);

    \draw[->,transform canvas={shift={(8pt,0pt)}}] (CLOCC) to node[pos=0.5,sloped,above]{Theorem\,\ref{t4}} (CLRA2);
    \draw[->,transform canvas={shift={(-8pt,0pt)}}] (CLRA2) -- node[midway,sloped]{\color{red}$\times$}(CLOCC);

  \end{tikzpicture}%
  }
  \caption{A schematic diagram for the relation between LRA and LOCC discrimination. In the figure LOCCD abbreviates the phrase \textit{LOCC distinguishable}.  All the solid directed arrows represents the results  established in the present work, while the dashed arrows are already known implications. The \textcolor{red}{red} cross-marks are used to depict that the corresponding implications do not hold true.}
  \label{fig:lra-locc}
\end{figure}

\section{Implication of LRA: change point problem}\label{s4}
We will now highlight an important consequence of the task LRA, in the change point problem of the multipartite quantum states. The quantum change point problem, in a nutshell, deals with an imperfect device intended to prepare a specific quantum state $\ket{\psi}$ and occasionally interrupted to produce a faulty preparation \cite{akimoto2011,sentis2016quantum,sentis2016,sentis2017,sentis2018,fanizza2023,nakahira2023}. From the perspective of classical information processing such a mutated preparation will depict its ultimate disagreement with the expected features, whenever it becomes orthogonal to the target state. This motivates us to consider a device, which is targeted to prepare a $d$-dimensional multipartite state $\ket{\psi}\in\mathcal{H}_d$ and erroneously mutated to one of its orthogonal states in $\mathcal{S}_{\text{mutd.}}:=\{\ket{\phi_1},\ket{\phi_2},\cdots,\ket{\phi_{d-1}}\}$. From now on, we will use $\ket{\psi}$ as the target state and $\ket{\phi_i}$ as an element of the mutated set. These states are shared among multiple parties, who are allowed to perform LOCC and estimate the exact change point, as recently also considered in \cite{bannerjee2024}.

Operationally, a local quantum change point problem involves an $M$-length sequence of quantum states. These states generated by a device targeted to prepare the state $\ket{\psi}$, however gets mutated to any of the states in $\mathcal{S}_{\text{mutd.}}$ from the $k^{\text{th}}$ iteration, where $k\leq M$. Hence, the problem of change point estimation can be identified with a triplet $(\ket{\psi},\mathcal{S}_{\text{mutd.}}, M)$ and the goal is to estimate the number $k\leq M$ from which the state transformed to mutated state. It is then straightforward to assign this problem with an instance of discriminating the set of quantum states $\mathcal{\tilde{S}}_M:=\{\ket{\xi_k}=\ket{\psi}^{\otimes k}\otimes\ket{\phi_l}^{\otimes(M-k)}\}_{k=0}^{M}$, where $\ket{\phi_l}\in\mathcal{S}_{\text{mutd.}}$. Now assuming that the target state can be uniformly mutated to any elements of the set $\mathcal{S}_{\text{mutd.}}$, we can replace $\ket{\phi_l}$ by the maximally mixed state $\frac{1}{d-1}\mathbb{I}_{\mathcal{S}}$ in the subspace spanned by $\mathcal{S}_{\text{mutd.}}$ and hence each $\ket{\xi_k}$ can be generalized to $\rho_k:=\ketbra{\psi}{\psi}^{\otimes k}\otimes(\frac{1}{d-1}\mathbb{I}_{\mathcal{S}})^{\otimes(M-k)}$.

An obvious question prior to the perfect local estimation of the change point $k\leq M$ is to identify whether the device has at all mutated or not. This in terms of the task LRA leads to the following observation.
\begin{observation}\label{obs1}
    The necessary condition for locally solving a change point problem $(\ket{\psi},\mathcal{S}_{\text{mutd.}},M)$ is that the set $\{\rho_k\}_{k=0}^M$ admits partial LRA for the question $\mathcal{Q}_{M}$, i.e., {\tt "whether the joint state is $\rho_M=\ketbra{\psi}{\psi}^{\otimes M}$ or not."}
\end{observation}
Theorem \ref{t2}, together with Observation \ref{obs1}, immediately directs to infer the following Corollary.
\begin{corollary}\label{coro2}
    In a change point problem $(\ket{\psi},\mathcal{S}_{\text{mutd.}},M)$, it is always possible to detect whether mutation has occurred or not if the target state $\ket{\psi}$ is product.
\end{corollary}
\begin{proof}
    Note that for any change point problem $(\ket{\psi},\mathcal{S}_{\text{mutd.}},M)$, all possible states in the set  $\mathcal{\tilde{S}}_{M}:=\{\rho_k\}_{k=0}^M$ are mutually orthogonal. Now, if the intended state $\ket{\psi}$ is product, Theorem \ref{t2} implies that $\ketbra{\psi}{\psi}^{\otimes M}$ is always possible to authenticate in the set $\mathcal{\tilde{S}}_{M}$ and hence the question $\mathcal{Q}_{M}$ can be answered perfectly. This completes the proof.
\end{proof}
  We now turn to the problem of local change point estimation for a set of multipartite quantum states and derive the necessary and sufficient conditions for its successful implementation. Naturally, an optimal strategy would involve a joint measurement across the local constituents of all $M$ copies of the multipartite state. However, such joint measurements typically require entangling operations and are therefore highly resource-intensive. Conversely, the simplest and most economical measurement to be performed on each of the subsystems, are adaptive strategies. Adaptive strategies are very useful in terms of resource consumption in any information theoretic tasks and has been studied extensively in various domain of Quantum information processing, such as state discrimination \cite{banik2019multicopy}, quantum channel discrimination \cite{Hayashi2009,Harrow2010,Pirandola2019} and quantum metrology and error estimation \cite{Pirandola2017,Cope2017}. In what follows, we introduce the notion of a local adaptive strategy in the context of change point problems. 
\begin{definition}
    (Local Adaptive Strategy) Given a change point problem $(\ket{\psi},\mathcal{S}_{\text{mutd.}},M)$, the local adaptive strategy is a sequential LOCC protocol in which each party measures the $j^{th}$ state in the sequence individually for each $j\in\{1,2,…,M\}$, with the goal of extracting maximal information about the change point. The measurement performed at the $j^{th}$ step is allowed to depend on the classical outcomes of all previous measurements, i.e., those from rounds $l<j$. Thus, the strategy, in each step, adapts dynamically based on the classical information accumulated from the previous rounds, while remaining restricted to LOCC.
\end{definition}
Under the action of such an adaptive LOCC implementable change point estimation problem, we obtain the following necessary and sufficient criterion for perfect accomplishment of the task. 
\begin{theorem}\label{t5}
    A change point problem $(\ket{\psi},\mathcal{S}_{\text{mutd.}},M)$ is locally solvable by adaptive strategy if and only if the set $\mathcal{S}=\{\ket{\psi}\}\cup\mathcal{S}_{\text{mutd.}}$ admits partial LRA at least for the question $\mathcal{Q}_{\psi}$.
\end{theorem}
\begin{proof}
    The `\textit{if}'-part follows directly by recalling the task of LRA for the question $\mathcal{Q}_{\psi}$. In each run, the question $\mathcal{Q}_{\psi}$ asks {\tt "whether the shared state is $\ket{\psi}$ or, not?"} and by adapting a LOCC implementable protocol $\mathcal{M}_{\text{LOCC}}$, the involved parties can answer this question perfectly. Now for every instances of the total $M$-length sequence they will apply the protocol $\mathcal{M}_{\text{LOCC}}$ and continued further to the next instance if the result is $1$ ({\tt `Yes'}), i.e., if the state is $\ket{\psi}$. Finally, if the result is $0$ ({\tt `No'}) in some $j\leq M$ instance, they can perfectly identify $j$ as the change point. Notably, they need not to proceed further after getting the $0$ outcome of the measurement $\mathcal{M}_{\text{LOCC}}$.

    The `\textit{only if}' direction is rather less straightforward, due to the fact that in general an adaptive strategy admits highly complex structure involving the information acquired from all of the previous rounds. However, we will argue that for estimating the change point perfectly, any local adaptive strategy can be replaced by a constant LOCC-implementable protocol on each of the $M$ sequences.

    The rationale behind this is that the source we are dealing with follows identical and independent distribution (i.i.d),  sampled over $\{\ket{\psi}\}\cup\mathcal{S}_{\text{mutd.}}$ upto the change point $k\leq M$ and thereafter a uniform i.i.d over $\mathcal{S}_{\text{mutd.}}$. Now, suppose in the $i^{\text{th}}$ round, with all the information acquired from the previous rounds, the distant parties are concluded by assigning non-zero probabilities over the set $\{\ket{\psi}\}\cup\{\ket{\phi_1},\cdots, \ket{\phi_l}\}$, where the set $\{\ket{\phi_1},\cdots, \ket{\phi_l}\}\subseteq\mathcal{S}_{\text{mutd.}}$ with $l\leq (d-1)$. However, due to the i.i.d structure the state in the $(i+1)^{\text{th}}$ round is again sampled over $\{\ket{\psi}\}\cup\mathcal{S}_{\text{mutd.}}$. Evidently, the information acquired from the $i^{\text{th}}$ round then has no implication to adapt a new LOCC protocol in the $(i+1)^{\text{th}}$ round. Following the same argument for all possible $i\leq k\leq M$, one can then conclude that the dynamical adaptive strategy is no way better than a constant LOCC protocol for each round.

    Now, suppose for the change point problem $(\ket{\psi},\mathcal{S}_{\text{mutd.}}, M)$, there is a adaptive LOCC protocol, which can perfectly identify the change point. Therefore, for each of the rounds the protocol must (atleast) identify whether the state is $\ket{\psi}$, or $\frac{1}{d-1}\mathbb{I}_{\mathcal{S}}$. Hence the same LOCC protocol can be adapted to locally authenticate the state $\ket{\psi}$, given from the set $\mathcal{S}=\{\ket{\psi}\}\cup\mathcal{S}_{\text{mutd.}}$. This completes the proof.
\end{proof}
We now state two important corollaries that characterize the structure of quantum states which either permit or prohibit local adaptive estimation in change point problems. The proofs of Corollaries \ref{coro3} and \ref{coro4} follow directly from Theorems \ref{t2} and \ref{t3}, respectively, in conjunction with Theorem \ref{t5}.
\begin{corollary}\label{coro3}
    A change point problem $(\ket{\psi},\mathcal{S}_{\text{mutd.}},M)$ is always possible to solve via local adaptive strategy if the state $\ket{\psi}$ is product.
    \end{corollary}

\begin{corollary}\label{coro4}
    A change point problem $(\ket{\psi},\mathcal{S}_{\text{mutd.}},M)$ can never be solved via local adaptive strategy if the set $\mathcal{S}:=\{\ket{\psi}\}\cup\mathcal{S}_{\text{mutd.}}$ forms an orthonormal basis in $\bigotimes_{k=1}^N \mathbb{C}^{d_k}$ consisting of all entangled states.
\end{corollary}
\section{Discussions}\label{s5}
In summary, we have explored local authentication of a classical index encoded in a quantum state, with respect to another independently asked index. The premise of the present task (LRA) can be seen as an initiative to highlight the role of an external classical input in context of LOCC state discrimination. Similar inclusion in classical communication paradigm \cite{wiesner1983conjugate, ambainis1999dense, ambainis2002dense}  gives rise to several distinctive features compared to the conventional one. Most remarkably introduction of inputs in communication unlocks the advantage of qudit communication \cite{Ambainis2009,Montina2011,Spekkens2009,Tavakoli2016,chowdhury2025gottesman}, which is otherwise equivalent to a cdit \cite{holevo1973bounds, frenkel2015classical}. Similarly, while in conventional LOCC discrimination, even product ensemble is known to exhibit nonlocality, introduction of inputs makes nonlocality truly sensitive to entanglement. This results in the absence of Nonlocality Without Entanglement like phenomena, making LRA the only known LOCC discrimination task that certifies entanglement.

We have further shown that impossibility of LRA reveals stronger notion of nonlocality compared to LOCC discrimination. Interestingly we have reported that product states---though authenticable locally---can activate conditional nonlocality in an otherwise local ensemble, a feat unattainable with entanglement. Moreover, we have established that full entangled basis in every dimension is conditionally nonlocal. Besides of its foundational interests, our work bears important implications in context of learning an erroneous device, targeted to prepare a particular multipartite quantum state. Specifically, the perfect estimation of the error in the device is possible whenever it is targeted to prepare product states. However, if the device produces all entangled states, both for the targeted and the mutated one, we show under adaptive LOCC such a device is forbidden to estimate.

A number of promising open research directions stem out as direct consequence of our findings. First, the type of the input questions considered here, is the simplest one, i.e., about the identity of the state. However, there could be various other possible input questions, viz., parity identification, subspace discrimination etc., which may differ by their local answerable complexities. Second, the present framework only considers the fully local operations along with unrestricted classical communication, both of which can be tuned further depending upon practical situations and causal constraints. More precisely, in the multipartite settings, the various forms of party accumulations, along with specific directions of classical communication may lead to identifications of different resources, viz., bi-separability, genuine entanglement etc. Third, due to the lack of concise mathematical description of LOCC, the question of local state discrimination is often asked in the paradigm of separable measurements \cite{yu2012four,duan2009distinguishability,bandyopadhyay2015limitations,cohen2015class} or positive partial transpose (PPT)-preserving operations \cite{yu2014distinguishability}. Similarly, exploring the status of the task LRA in such broader classes of operations can be noteworthy. Fourth, it is instructive that for the two-qubit case, our results point out the maximum cardinality of a locally authenticable set with all entangled members is exactly \textit{three} (see Theorem \ref{t1} and \ref{t3}). However, the same in general bipartite settings is not known. 
Finally, we are hopeful to have important implications of the local inaccessibility to authenticate the randomized quantum states in the cryptographic context, including the stronger notions of data hiding and secret sharing.

\section{Acknowledgments}
We thank Lu Yingying for their comment in the earlier version of Proposition 2. We thankfully acknowledge fruitful discussions with Guruprasad Kar, Manik Banik and Amit Mukherjee. Tamal Guha is supported by Hong Kong Research Grant Council (RGC) through Grant No. 17307719 and 17307520. SRC acknowledges support from University Grants Commission(reference no. 211610113404), India.

    \bibliographystyle{unsrt}
\bibliography{bibliography}

\begin{thebibliography}{10}
\bibitem{comment}
In the present scenario of unrestricted classical communication, distribution of classical inputs among the spatially separated receivers is equivalent to its broadcasting.
\end{thebibliography}
\end{document}